%% file: full_v.tex
\documentclass[11pt]{article}

\usepackage{fullpage}
\usepackage{amsmath, amsthm, amssymb}
\usepackage{hyperref}
\usepackage[noend]{algpseudocode}
\usepackage{algorithmicx, algorithm}
\usepackage{supertabular}
\usepackage{graphicx,accents}
\usepackage{etoolbox}
\usepackage{subcaption}

\begin{document}

\makeatletter
\DeclareRobustCommand{\cev}[1]{%
	\mathpalette\do@cev{#1}%
}
\newcommand{\do@cev}[2]{%
	\fix@cev{#1}{+}%
		\reflectbox{$\m@th#1\vec{\reflectbox{$\fix@cev{#1}{-}\m@th#1#2\fix@cev{#1}{+}$}}$}%
		\fix@cev{#1}{-}%
}
\newcommand{\fix@cev}[2]{%
	\ifx#1\displaystyle
		\mkern#23mu
		\else
		\ifx#1\textstyle
		\mkern#23mu
		\else
		\ifx#1\scriptstyle
		\mkern#22mu
		\else
		\mkern#22mu
		\fi
		\fi
		\fi
}
\makeatother

\newcommand{\ppr}[2]{\pi(#1, #2)}
\newcommand{\pr}[1]{Pr(#1)}
\newcommand{\List}[1]{[\,#1\,]}
\newcommand{\set}[1]{\{#1\}}
\newcommand{\rmax}{\varepsilon}
\newcommand{\ninf}{l_{\infty}}
\newcommand{\abs}[1]{|#1|}
\newcommand{\seq}[1]{\langle#1\rangle}
\newcommand{\ub}[2]{#1^{(#2)}}
\newcommand{\norm}[1]{\lVert#1\rVert}
\newcommand{\eqplus}{\mathrel{+}=}
\newcommand{\eqminus}{\mathrel{-}=}
\newcommand{\dout}{\textnormal{\tiny {\sc out}}}
\newcommand{\din}{\textnormal{\tiny {\sc in}}}
\newcommand{\tn}[1]{\textnormal{\tiny {#1}}}
\newcommand{\fvec}{\vec{P}}
\newcommand{\gvec}{\vec{R}}
\newcommand{\bvec}{\cev{P}}
\newcommand{\rvec}{\cev{R}}
\newcommand{\pivec}{\vec{\pi}}
\newcommand{\phivec}{\vec{\Phi}}
\newcommand{\shrink}[1]{{\small{#1}\normalsize}}
\newcommand{\E}[1]{\operatorname{E}\,[\,#1\,]}
\newcommand{\K}{k}
\newcommand{\N}{\mathit{N}}

\newcommand{\dbar}{\overline{d}}
\newcommand{\inlineheading}[1]{\textbf{#1}}

\newenvironment{proofof}[1]{\smallskip\noindent{\em Proof of #1:}} {\hspace*{\fill}\par}
\algrenewcommand\alglinenumber[1]{\tiny #1:}
\algnewcommand\algorithmicinput{\textbf{INPUT:}}
\algnewcommand\Input{\item[\algorithmicinput]}

\newtheorem{theorem}{Theorem}
\newtheorem{lemma}[theorem]{Lemma}
\newtheorem{proposition}[theorem]{Proposition}
\newtheorem{corollary}[theorem]{Corollary}
\newtheorem{definition}{Definition}

\newtoggle{fullpaper}
\toggletrue{fullpaper}

\title{Approximate Personalized PageRank on Dynamic Graphs}

\author{
	Hongyang Zhang\\
	Department of Computer Science\\
	Stanford University\\
	\texttt{hongyang@cs.stanford.edu}
	\and
	Peter Lofgren\\
	Department of Computer Science\\
	Stanford University\\
	\texttt{plofgren@cs.stanford.edu}
	\and
	Ashish Goel\\
	Department of Management Science and Engineering\\
	Stanford University\\
	\texttt{ashishg@stanford.edu}\\
}
\date{}
\maketitle

\input{abstract}

\newpage

\input{intro}

\input{prelim}

\input{alg}

\input{experiment}

\section{Future work}

In this work we presented an approach to obtain and analyze dynamic
local push algorithms.
We mention several questions that are worth further study.
In the submitted version we proposed an algoirthm that only uses random walk samples
to compute personalized PageRank between a source and a target node for undirected graphs,
but mistakenly claimed that its time complexity is the same as the state of the art~\cite{waw:bid}.
It is interesting to see if such an algorithm exists or not.
Secondly, in practice, some graphs cannot fit into one machine, and it would
be interesting to see how push algorithms perform versus random
walk algorithms in such a distributed setting.
Another question is whether one can obtain an error analysis of the forward push
algorithm on a directed graph.
Understanding such an issue can be helpful for applying it in practice as well.

\paragraph{Acknowledgement}

Research supported by the DARPA GRAPHS program via grant
FA9550-12-1-0411, and by NSF grant 1447697.

\bibliographystyle{plain}
\bibliography{rf,PPR-refs}

\end{document}

%% file: abstract.tex
\begin{abstract}

We propose and analyze two algorithms for maintaining approximate Personalized PageRank (PPR) vectors on a dynamic graph, where edges are added or deleted. Our algorithms are natural dynamic versions of two known local variations of power iteration.  One, Forward Push, propagates probability mass forwards along edges from a source node, while the other, Reverse Push, propagates local changes backwards along edges from a target.  In both variations, we maintain an invariant between two vectors, and when an edge is updated, our algorithm first modifies the vectors to restore the invariant, then performs any needed local push operations to restore accuracy.

For Reverse Push, we prove that for an arbitrary directed graph in a random edge model, or for an arbitrary undirected graph, given a uniformly random target node $t$, the cost to maintain a PPR vector to $t$ of additive error $\varepsilon$ as $k$ edges are updated is $O(k + \overline{d} / \varepsilon)$, where $\overline{d}$ is the average degree of the graph.  This is $O(1)$ work per update, plus the cost of computing a reverse vector once on a static graph.  For Forward Push, we show that on an arbitrary undirected graph, given a uniformly random start node $s$, the cost to maintain a PPR vector from $s$ of degree-normalized error $\varepsilon$ as $k$ edges are updated is $O(k + 1 / \varepsilon)$, which is again $O(1)$ per update plus the cost of computing a PPR vector once on a static graph. 

\end{abstract}

%% file: intro.tex
\section{Introduction}

Personalized PageRank (PPR) models the relevance of nodes in a network from the point of view of a given node.  It has applications in search \cite{Jeh2003,03:Hav}, friend recommendations \cite{backstrom2011supervised,gupta2013wtf}, community detection \cite{yang2012defining, focs:chung}, video recommendations \cite{baluja2008video}, and other applications.  Because PPR is expensive to compute at query time, several authors have proposed pre-computing it for each user and storing it \cite{Jeh2003, berkhin2006bookmark,chakrabarti2007dynamic}. However, in practice graphs are dynamic, for example on a social network users are constantly adding new edges to the network, so we need a method of updating pre-computed PPR values.  In this work, we propose two new algorithms for updating pre-computed PPR vectors and give the first rigorous analysis of their running time.

There are four main algorithms for computing PPR \cite{thesis:lofgren}, of which only one has an analysis for dynamic graphs prior to our work.  The first is power iteration \cite{99:pr}, but it is very slow, requiring $\Omega(m)$ time per user (where $m$ is the number of edges), so it can't be used efficiently for PPR even on static graphs.  The second, Forward Push \cite{berkhin2006bookmark,focs:chung}, is a local variation of power iteration which starts from the source user (the node whose point of view we take) and pushes probability mass forwards along edges.  Berkin \cite{berkhin2006bookmark} proposed pre-computing Forward Push vectors from many source nodes, and Ohsaka et~al.~\cite{kdd:dynPr} proposed an algorithm for updating it on dynamic graphs, however no past work has analyzed the running time for maintaining Forward Push.  The third, Reverse Push \cite{Jeh2003,focs:chung}, is an alternative variation of power iteration which starts at each target node and pushes values backwards along edges to improve estimates.  Jeh and Widom \cite{Jeh2003} and Lofgren et~al.~\cite{wsdm:residualSampling} propose pre-computing Reverse Push vectors (or a variation on them), to enable efficient search, but no past work has proposed an algorithm for updating Reverse Push vectors on dynamic graphs.  Finally, a fourth method of computing PPR is Monte-Carlo \cite{Avrachenkov2007,Fogaras2005}, and for that algorithm Bahmani et~al.~\cite{vldb:walkUpd} give an efficient algorithm for updating it on dynamic graphs, with running time analysis in a random edge arrival order model.  In experiments, Ohsaka et~al.~\cite{kdd:dynPr} find that when maintaining very accurate PPR vectors, Monte-Carlo is slower than Forward-Push, motivating the analysis of updating Forward Push.

Our contribution is the first running time guarantees for updating Forward Push on undirected graphs and for updating Reverse Push on directed as well as undireted graphs.  The analysis is challenging because a new edge at one node can cause that node to push, which can cause a cascade of other nodes to push.  Understanding and bounding the size of this cascade required a novel amortized analysis.
We allow for a worst case graph, but we assume that edge updates arrive in a random order to better capture performance in practice.  The same assumption was used in the analysis of Monte Carlo \cite{vldb:walkUpd} where the bounds it leads to were found to model the running time in practice on Twitter.  Alternatively, for undirected graphs we prove a worst-case amortized running time.  In addition, for undirected graphs, we strengthen the analysis of Monte-Carlo on dynamic graphs \cite{vldb:walkUpd} to a worst-case edge arrival order.

Our algorithms are simple to implement, and we present two experiments for forward push.  One insight from our theoretical analysis is a different way of updating Forward Push values than the previously proposed method \cite{kdd:dynPr}, and we find that our variation is 1.5-3.5 times faster than that method without sacrificing accuracy.
In the second experiment, we evaluate forward push for the problem of finding the top K highest personalized PageRank
from a source node.
We compare to the Monte-Carlo method of Bahmani et al.~\cite{vldb:walkUpd} and found that
the forward push is 4.5-12 times more efficient in storage compared to Monte-Carlo,
and is able to do edge update 1.5 - 2.5 times faster.

The main idea of our algorithms is that both Forward Push and Reverse Push maintain an invariant between a pair of vectors, and when an edge is added or deleted, we first restore the invariant with an efficient change to the vectors, then perform local push operations as needed to control the error. 
One limitation of our analysis is that it applies to the pure versions of Forward Push and Reverse Push, while some  prior work \cite{Jeh2003,berkhin2006bookmark} proposes combining Forward or Reverse push vectors together, but we believe our analysis would extend to that case.  It would also be interesting to extend our analysis adapt the Personalized PageRank search index of \cite{wsdm:residualSampling} to dynamic graphs.

\paragraph{Results}
For the Reverse Push algorithm, we show that on a worst-case graph, for a uniform random target, the expected time required to maintain PPR estimates at accuracy $\rmax$ to that target as $k$ edge updates arrive in a random order is $O(k + k/(n \rmax) + \dbar / \rmax)$.  Here $\rmax$ is the desired additive error of the estimates, and $\dbar$ is the average degree of nodes in the graph.  Since a random PPR value is $1/n$, values smaller than $1/n$ are not very meaningful.  Hence typically $\rmax = \Omega(1/n)$ and our running time is $O(k + \dbar / \rmax)$.  Since $\dbar / \rmax$ is the expected time required to compute Reverse Push from scratch \cite{waw:chung, thesis:lofgren}, we see that our incremental algorithm can maintain estimates after every edge update using $O(1)$ time per update and the time required to compute a single Reverse Push vector.

For the Forward Push algorithm, we show that on an arbitrary undirected graph and arbitrary edge arrival order, for a uniform random source node $s$, the worst-case running time to maintain a PPR vector from that source node is $O(k + k/(n\rmax) + 1/ \rmax)$.  Here $\rmax$ is a bound on the degree-normalized error (which we define later).  Typically $\rmax = \Omega(1/n)$, so this is $O(k + 1 / \rmax)$.  The cost of computing such a PPR vector from scratch is $O(1 / \rmax)$ \cite{focs:chung} so we again see that the cost to maintain the PPR vector over $k$ updates is $O(1)$ per update plus the cost of computing it once from scratch.

One observation that follows from our work is that for Monte-Carlo, for an arbitrary undirected graph and arbitrary edge arrival order, the time required to maintain $r$ random walks from every source node over $k$ arrivals is $O(r k)$.
This contrasts with a worst-case example constructed on directed graphs
where the total running time to maintain these walks grows super-linear in $k$ \cite{rw:worstcase}.
Our observation suggests a weaker yet still meaningful bound without the random edge assumption made in Bahmani et al.~\cite{vldb:walkUpd}
for Monte-Carlo methods on undirected graphs.

%% file: prelim.tex
\section{Preliminaries}

Let $G = (V, E)$ be an unweighted directed graph.
Let $A$ denote the adjacency matrix of $G$ and
let $D$ denote the diagonal matrix representing the outdegrees of all vertices in $G$.
For a vertex $v$, let $\N^{\dout}(v)$ denote the set of out neighbors of $v$
and let $\N^{\din}(v)$ denote the set of in neighbors of $v$.
The personalized PageRank vector $\pivec_s$ for a source node $s$ is defined as the unique solution of the following linear system (\cite{99:pr}):
\begin{equation}\label{eq:ppr}
	\pivec_s = \alpha \cdot \vec e_s + (1 - \alpha) \cdot A^\intercal D^{-1} \pivec_s
\end{equation}
where $\alpha$ (a.k.a. the teleport probability) is a constant between $0$ and $1$,
and $\vec {e_s}$ is the indicator vector with a single nonzero entry of $1$ at $s$.
For a pair of vertices $s$ and $t$ on $G$,
we will use $\ppr s t$ to denote the personalized PageRank from $s$ to $t$.
This linear algebraic definition of personalized PageRank is equivalent to simulating a random walk.
Start from the source node $s$,
with probability $(1 - \alpha)$, go to a uniformly chosen neighbor of the current node,
or with probability $\alpha$ stop at the current node:
$\ppr s t$ is the probability that a random walk from $s$ stops at $t$~\cite{99:pr}.
In the above definitions, we assumed that the outdegree of every vertex is at least one for simplicity.
If a vertex has no out neighbors, then the random walk transitions to $s$ with probability $(1 - \alpha)$
(cf. Gleich~\cite{survey:gleich} for other ways to handle dangling nodes).

\subsection{The edge arrival model}\label{sec:model}

In the dynamic edge arrival model, we start with an initial graph and there is a sequence of edge updates one by one.
Let $G_0 = (V_0, E_0)$ denote the initial graph.
Let $\K$ denote the number of edge updates.
When the $i$-th edge $e_i = (u_i, v_i)$ arrives,
if $e_i$ is already in $G_{i-1}$, then it will be deleted; else it will be added to $G_{i-1}$.
Let $G_i = (V_i, E_i)$ denote the updated graph.
Note that $V_i$ can differ from $V_{i-1}$ by at most two vertices.
Let $d^{\dout}_i(u)$ denote the outdegree of $u \in V$ on $G_i$ and $d^{\din}_i(u)$ denote its indegree.
Also let $D_i$ denote the diagonal matrix of the outdegrees of $V_i$.
Let $\pi_i(s, t)$ denote the personalized PageRank from $s$ to $t$ on $G_i$.
Let $n = \abs V_0$ denote the number of vertices  and $m = \abs {E_0}$ denote the number of edges in the initial graph $G_0$.
We will analyze the following two edge arrival models:
	
	\paragraph{Random edge permutations of directed graphs~\cite{vldb:walkUpd}}
	To introduce this graph model in our notation,
	consider a uniformly random edge permutation of $E_{\K}$.
	Then $G_0$ is the subgraph consisting of the first $m$ edges in the permutation.
	And $e_i$ is defined as the $(m+i)$-th edge in this permutation,
	for $i = 1,\dots,\K$.
	In a uniformly random edge stream, the personalized PageRank vector does not change by ``too much''
	after an edge update~\cite{vldb:walkUpd}.
	This intuition is not true in the worst case ---
	there exists a sequence of edge insertions such that for a single source node,
	the aggregate $l_1$ difference between the personalized
	PageRank vector before and after every edge insertion
	grows superlinear in the number of of nodes~\cite{rw:worstcase}.

	\paragraph{Arbitrary edge updates of undirected graphs} For undirected graph, we observe a few 
	interesting properties for personalized PageRank, which will be used later.
	The first one exploits the fact that on an undirected graph,
	a random walk can be in either direction.

		\begin{proposition}[cf. Lemma $1$ in Lofgren et al.~\cite{waw:bid}]\label{prop:reversible}
			Let $G$ be an undirected graph. Let $s$ and $t$ be two vertices of $G$.
			Then $\ppr s t \times d(s) = \ppr t s \times d(t)$.
		\end{proposition}
		Here we dropped the superscript on $d$ since there is no direction for undireced edges.
		\begin{proposition}\label{prop:updateBound}
			Let $G = (V, E)$ be an undirected graph and let $t$ be a vertex of $V$, then 
			$ \sum_{x \in V} \ppr x t / {d(t)} \le 1$.
		\end{proposition}

		\begin{proof}
			Note that,
			\[\sum_{x \in V} \frac {\ppr x t} {d(t)}
				= \sum_{x \in V} \frac {\ppr t x} {d(x)} \le \sum_{x\in V} \ppr t x = 1. \]
			where we used Proposition~\ref{prop:reversible}.
		\end{proof}

\subsection{Local push algorithms}
The random walk interpretation of personalized PageRank leads to the following general
class of local computation algorithms:
start with all the probability mass on the source node of the graph,
and iteratively push this mass out to neighboring nodes,
getting progressively better approximations.
In these ``local push'' algorithms, we typically maintain a ``residual'' at every node,
which is mass that has been received but not yet been pushed out from that node,
as well as an ``estimate'', which is mass that has been both received and pushed out
(after accounting for the teleport probability and the outdegrees), we will now
precisely describe how this local approach can be used in the forward and reverse
direction.

\subsubsection{Forward push}

Given a source node $s$, the forward push algorithm maintains an estimate $\fvec(s, t)$ of $\ppr s t$ for each target $t \in V$.
It also maintains a residual $\gvec(s, t)$ for $t$.
Let $\fvec(s) = \fvec(s, \cdot)$ denote the the vector of estimates and
$\gvec(s)$ the vector of residuals.
The estimates and residuals satisfy the following
invariant property (cf. Section $3$ in Lofgren~\cite{focs:chung}):
	\begin{equation}\label{eq:fwdReisdualForm}
		\ppr s t = \fvec(s,t) + \sum_{x \in V} \gvec(s, x) \times \ppr x t, \forall t \in V.
	\end{equation}
Algorithm~\ref{alg:lpFwd} described below is a variant of the classic forward local push algorithm \cite{focs:chung},
the difference being that we added negative residuals.
As we will see later on, an edge arrival can result in negative residuals ---
the edge update affects the amount of residuals a vertex have pushed out.
During a forward push iteration for vertex $u$ (Step~\ref{step:fwdPush} in Algorithm~\ref{alg:lpFwd}),
an $\alpha$ fraction of $u$'s residual will be added to $u$'s estimate.
For the rest of $(1-\alpha)$ fraction, each out neighbor of $u$ receives an equal proportion
of $(1-\alpha) / d^{\dout}(u)$.
Algorithm~\ref{alg:lpFwd} repeatedly perform forward push iterations,
first for vertices with positive residuals and then for vertices with negative residuals,
until for every vertex, its residual divided by the outdegree is within $-\rmax$ and $\rmax$.
It's clear that while we work on vertices with negative residuals,
no vertices whose residual is positive and below $\rmax$ will increase above $\rmax$.

\newcommand{\alglpFwd}{{\sc ForwardLocalPush}}
\begin{algorithm}
	\small
	\caption{\alglpFwd}
	\label{alg:lpFwd}
	\begin{algorithmic}[1]
		\Input ($s, \fvec(s), \gvec(s), G, \rmax$)
		\While {$\max_u \frac {\gvec(s, u)} {d^{\dout}(u)} > \rmax$}
			\State FwdPush(u)
		\EndWhile
		\While {$\min_u \frac {\gvec(s, u)} {d^{\dout}(u)} < -\rmax$}
			\State FwdPush(u)
		\EndWhile\\
		\Return ($\fvec(s), \gvec(s)$)
		\vspace{0.05in}
		\Procedure{FwdPush}{$u$} \label{step:fwdPush}
			\State $\fvec(s, u) \eqplus \alpha \times \gvec(s, u)$
			\For {$v \leftarrow u$}
				\State $\gvec(s, v) \mathrel{+}= (1 - \alpha) \times \gvec(s, u) / d^{\dout}(u)$
			\EndFor
			\State $\gvec(s, u) \gets 0$
		\EndProcedure
	\end{algorithmic}
\end{algorithm}

\subsubsection{Reverse push}\label{sec:revPush}

Given a target node $t$,
the reverse push algorithm maintains an estimate $\bvec(s, t)$ of $\pi(s, t)$, for every node $s$ in $G$.
It also maitains a residual $\rvec(s, t)$ for $s$.
Let $\bvec(t) = \bvec(\cdot, t)$ denote the vector of estimates and $\rvec(t) = \rvec(\cdot, t)$ for the vector of residuals.
Similar to forward push, estimates and residuals satisfy an invariant property \cite{waw:chung, kdd:fastppr}:
\begin{equation}\label{eq:residualForm}
	\ppr s t = \bvec(s, t) + \sum_{x \in V} \ppr s x \times \rvec(x, t), \forall s \in V.
\end{equation}
In Algorithm~\ref{alg:lp} below, we've also added negative residuals.
A reverse push iteration on vertex $u$ works backward (step~\ref{step:revPush}):
an $\alpha$ fraction of $u$'s residuals goes to $u$'s estimate;
to push back the other $1-\alpha$ fraction to an in neighbor $v$ of $u$, it takes into account
the outdegree of $v$, hence the amount is the residual amount times $(1-\alpha) / d^{\dout}(v)$.
Algorithm~\ref{alg:lp} keeps pushing residuals back from each vertex,
until the residual of every vertex is within $-\rmax$ and $\rmax$.
When this happens, it is guaranteed that $\abs {\ppr s t - P(s, t)} \le \rmax$ (cf. Theorem $1$ in Anderson et al.~\cite{waw:chung}).

\newcommand{\alglp}{{\sc ReverseLocalPush}}
\begin{algorithm}
	\caption{\alglp}\label{alg:lp}
	\small
	\begin{algorithmic}[1]
		\Input ($t, \bvec(t), \rvec(t), G, \rmax$)
		\While {$\max_u \rvec(u, t) > \rmax$} \label{step:pushPos1}
			\State RevPush(u) \label{step:pushPos2}
		\EndWhile
		\While {$\min_u \rvec(u, t) < -\rmax$} \label{step:pushNeg1}
			\State RevPush(u)	\label{step:pushNeg2}
		\EndWhile\\
		\Return ($\bvec(t), \rvec(t)$)
		\vspace{0.05in}
		\Procedure{RevPush}{$u$}\label{step:revPush}
			\State $\bvec(u, t) \eqplus \alpha \times \rvec(u, t)$
			\For {$v \rightarrow u$}
				\State $\rvec(v, t) \eqplus (1 - \alpha) \times \rvec(u, t) / d^{\dout}(v)$
			\EndFor
			\State $\rvec(u, t) \gets 0$
		\EndProcedure
	\end{algorithmic}
\end{algorithm}

\iftoggle{fullpaper}{
\subsection{Random walks}\label{sec:rw}
A well known method in dynamic settings uses random walks
to maintain the personalized PageRank vector.
The algorithm together with an analysis is first introduced by Bahmani et al.~\cite{vldb:walkUpd}.
The observation is that an edge update will only affect walks that visited the vertices of the edge.
As an example, suppose that edge $u \rightarrow v$ is inserted.
If a random walk visits $u$, it should go through edge $u \rightarrow v$ with probability 
$1 / (d^{\dout} + 1)$.
Therefore, the random walk can be reroute from $u$ by generating a new walk segment from $u$.
We repeat the rerouting procedure for each visit at $u$, unless the random walk has been
rerouted already.
The expected number of random walks that needs to be rerouted is given below.

\begin{proposition}[Bahman et al.~\cite{vldb:walkUpd}]\label{prop:walkUpd}
	Let $G = (V, E)$ be a directed graph and $s$ be a vertex of the graph.
	Let $w$ be a random walk from $s$.
	Now suppose that an edge $u \rightarrow v$ is inserted to $G$.
	Then the probability that $w$ needs to be updated is at most $\frac {\pi(s, u)} {\alpha \times (d^{\dout}(u) + 1)}$.
\end{proposition}
}

%% file: alg.tex
\section{Dynamic local push algorithms}

Given a node, we can run push algorithms to initialize its data structures that include the estimates and residuals.
If an edge is inserted/deleted, 
we can obtain a new pair of estimates and residuals,
by adjusting them locally at the updated edge using the invariant equations. 
Since this adjustment can create residuals that gets above $\rmax$ or below $-\rmax$,
we then invoke Algorithm~\ref{alg:lpFwd} (or Algorithm~\ref{alg:lp}) to push residuals from such nodes.
While the idea is simple, the analysis of the number of push iterations is quite intricate.
For the rest of this section, we describe an approach to obtain dynamic local push algorithms and show
how to analyze their running time.
We will introduce a dynamic reverse push algorithm in the first part.
Then we will introduce a dynamic forward push algorithm in the second part.
Finally we will consider how to extend the algorithms to handle node arrivals and other issues.

\subsection{Reverse push}

Let $t$ be a (target) vertex of $G$.
Let $\rmax$ be a parameter between $0$ and $1$.
Our goal is to maintain a pair of estimates and residuals, denoted by $\bvec(t)$ and $\rvec(t)$,
such that $\norm {\rvec(t)}_{\infty} \le \rmax$.
As mentioned in Section~\ref{sec:revPush}, this will ensure that
$\abs {\bvec(s, t) - \ppr s t} \le \rmax$, for any node $s \in V$.
We start with an equivalent formulation of Equation~\ref{eq:residualForm}.

\begin{lemma}\label{lem:invariance}
Equation \eqref{eq:residualForm} implies
	\shrink{\begin{equation*}\label{eq:invariance}
		\bvec(s, t) + \alpha \cdot \rvec(s, t) =
		\sum_{x \in N^{\dout}(s)} (1 - \alpha) \times \frac {\bvec(x, t)} {d^{\dout}(s)} + \alpha \times \mathbf 1_{s=t},~\forall s \in V,
	\end{equation*}}
and vice versa.
\end{lemma}
\begin{proof}
	Let $\pivec^t = \pi(\cdot, t)$ denote the vector with personalized PageRank from every node of $G$ to $t$.
	Let $\Pi = (I - (1-\alpha) D^{-1} A) / \alpha$.
	It follows from the work of Haveliwala~\cite{03:Hav} that
	$\Pi$ is invertible,
	and $\pivec^t = \Pi^{-1} \cdot \vec e_t$.
	This implies that the $(s, t)$-th entry of $\Pi^{-1}$ is equal to $\ppr s t$,
	since the $s$-th entry of $\pivec^t$ is $\ppr s t$,
	Now, we can write Equation \eqref{eq:residualForm} in vector form:
	\shrink{\begin{align*}
										& ~\pivec^t = \bvec(t) + \Pi^{-1} \cdot \rvec(t) \\
		\Leftrightarrow&~ \Pi \cdot \pivec^t = \Pi \cdot \bvec(t) + \rvec(t) \\
		\Leftrightarrow&~ \vec e_t = \Pi \cdot \bvec(t) + \rvec(t) \\
		\Leftrightarrow&~ \bvec(t) + \alpha \times \rvec(t) = (1-\alpha) D^{-1}A \cdot \bvec(t) + \alpha \times \vec e_t
	\end{align*}}
\end{proof}

Hence it follows if one can maintain the equivalent invariant
in Lemma~\ref{lem:invariance},
then one obtains a feasible pair of estimates and residuals
for the updated graph.
And the question becomes how to maintain the invariant between the estimate and the residual vectors.
We will only do it for edge insertions,
and it's not hard to work out the details for edge deletions using the same approach.
Suppose that an edge $u \rightarrow v$ is inserted into $G$.
Then the only vertex that doesn't satisfy the invariant in Lemma~\ref{lem:invariance} is $u$.
In this case, it suffices to update $\rvec(u, t)$ without changing any entries of $\bvec(t)$,
because the only variable that has been changed is $u$'s outdegree.
The precise formula is obtained
by calculating $u$'s correct new residual minus its old residual amount.
To simplify the expression below, we take out a common factor of $1 / \alpha$,
which comes from dividing the $\alpha$ multiplier of $\rvec(u, t)$ as in Lemma~\ref{lem:invariance}:
\shrink{\begin{align*}
	&	\left(\sum_{x \in \N^{\dout}(u)} \frac{(1-\alpha) \times \bvec(x, t)} {d^{\dout}(u)+1} \right)
		+ \frac {(1-\alpha) \times \bvec(v, t)} {d^{\dout}(u)+1} \\
	&	+ \alpha \times \mathbf 1_{u=t} - \bvec(u, t) \\
	& - \left(\sum_{x \in \N^{\dout}(u)} \frac{(1-\alpha) \times \bvec(x, t)} {d^{\dout}(u)} \right)
		- \alpha \times \mathbf 1_{u=t} + \bvec(u, t)\\
	& = \frac {(1-\alpha)\times \bvec(v, t)} {d^{\dout}(u) + 1} -
		\sum_{x \in \N^{\dout}(u)} \frac {(1-\alpha)\times \bvec(x, t)} {d^{\dout}(u) \times (d^{\dout}(u) + 1)} \\
	&= \frac {(1-\alpha)\times \bvec(v, t)} {d^{\dout}(u) + 1} -
		\frac {\bvec(u, t) + \alpha \times \rvec(u, t) - \alpha \times \mathbf 1_{u=t}} {d^{\dout}(u) + 1}
\end{align*}}%
Hence the insertion procedure in Algorithm~\ref{alg:lpUpdate} described below correctly generates
a pair of estimates and residuals for the updated graph.
Similarly for deletions.
It is worth mentioning that we haven't described how to handle nodes whose indegree is zero (also known as dangling nodes):
since they do not have in neighbors, their residuals cannot be pushed out.
We put off this issue until Section~\ref{sec:discussions}.
Another issue is, if one works with undirected graphs, one needs to apply the insert/delete procedure
for both direction of an edge.

\newcommand{\alglpUpd}{{\sc UpdateReversePush}}
\begin{algorithm}
	\small
	\caption{\alglpUpd}
	\label{alg:lpUpdate}
	\begin{algorithmic}[1]
		\Input ($t, \bvec(t), \rvec(t), u, v, G, \rmax$)
		\Require $G$ is a directed graph.
			Let $u \rightarrow v$ be the previous edge update, 
			with $G$ being the updated graph.
		\State Apply Insert/Delete to $\bvec(t)$ and $\rvec(t)$.\label{step:upd}
		\State \Return \alglp$(t, \bvec(t), \rvec(t), G, \rmax)$. \label{step:invokeRev}
		\vspace{0.05in}
		\Procedure{Insert}{$u, v$} \label{step:insertBk}
			\State $\rvec(u, t) \eqplus \frac {(1-\alpha) \times \bvec(v, t)
					- \bvec(u, t) - \alpha \times \rvec(u, t) + \alpha \times \mathbf 1_{u=t}} {d^{\dout}(u)} \times \frac 1 {\alpha} $ \label{step:ins}
		\EndProcedure
		\vspace{0.05in}
		\Procedure{Delete}{$u, v$}
			\State $\rvec(u, t) \eqminus \frac {(1-\alpha) \times \bvec(v, t)
					- \bvec(u, t) - \alpha \times \rvec(u, t) + \alpha \times \mathbf 1_{u=t}} {d^{\dout}(u)} \times \frac 1 {\alpha} $ \label{step:del}
		\EndProcedure
		\newline
	\end{algorithmic}
\end{algorithm}

For the rest of this section, our goal is presenting an analysis of Algorithm~\ref{alg:lpUpdate}.
Based on previous work of Bahmani et al.~\cite{vldb:walkUpd},
it is not difficult to observe that the updated amount of residuals from step \ref{step:ins} and \ref{step:del})
should be small in expectation in a random edge permutation model.
However, one can observe that there are already a lot of nonzero residuals in $\rvec(t)$.
While these residuals have all been reduced below $\rmax$,
it gets unclear if they couple with the updated amount of residuals.
Our intuition is to solve this problem with amortized analysis.

\begin{theorem}\label{thm:revPush}
	Let $\seq {G_i = (V_i, E_i)}$ be a sequence of $\K+1$ graphs such that each graph
	is obtained from the previous graph with one edge update.
	Let $\bar d$ denote the average degree of $G_0$.
	Let $t$ be a random vertex of $G_0$.
	Then the total running time of maintaining a reverse push solution $\bvec_i(t)$
	for each graph $G_i$ such that $\abs {\bvec_i(s, t) - \pi_i(s, t)} \le \rmax$,
	for any $s \in V_i$,
	using Algorithm~\ref{alg:lp} and~\ref{alg:lpUpdate},
	is at most $O(\K / \alpha + \K/(n\rmax\alpha^2) + \bar d / (\alpha\rmax))$ for the following two dynamic graph models:
		\begin{itemize}
			\item Arbitrary edge updates of an undirected graph.
			\item Random edge permutation of a directed graph,
				under the assumption that the indegree of every vertex of $G_0$ is at least one,
				and no new nodes are added during the edge arrivals.
		\end{itemize}
\end{theorem}

Remark: the assumptions under the case of directed graphs is because our proof
does not deal with issues of bounding the running time of handling dangling nodes,
and this can happen a node of $G_0$ has no in neighbors,
or when a new node arrives, thus creating a node without in neighbors.

We prove this theorem in three steps.
First of all, we derive a bound on the running time of Algorithm~\ref{alg:lp}.
Secondly, we present a bound on the running time of Algorithm~\ref{alg:lpUpdate}
and derive the total running time.
Finally, we bound the total running time using properties of the graph model.
Due to space limit, we only include the proofs for the case of undirected graphs,
and refer the reader to the full version for proofs of the other case.

We start with the running time of Algorithm~\ref{alg:lp}.
Lemma~\ref{lem:initCost} below is a Corollary of Theorem $2$ in the work of Lofgren et al.~\cite{kdd:fastppr},
the difference being that we subtracted a part that corresponds to the total cost of pushing all the remaining residuals,
since these residuals have not yet been pushed out.
As we will see later on, this part naturally arises in dynamic settings.
Let
	\shrink{\begin{align*}
		(\bvec_0(t), \rvec_0(t)) \triangleq \mbox{\alglp}(t, \vec e_t, \mathbf 0, G_0, \rmax),
	\end{align*}}%
and
	\begin{align}\label{eq:phi}
		& \phivec_i(x) \triangleq \sum_{s \in V_i} d^{\din}_i(s) \times \pi_i(s, x), \mbox{ and} \\
		& \phivec_i \triangleq \phivec_i(\cdot), \mbox{ its vector form}
	\end{align}

\begin{lemma}\label{lem:initCost}
	The running time of Algorithm \ref{alg:lp} is at most:
	\shrink{
	\begin{equation}\label{eq:initCost}
		\frac {\phivec_{\tn 0}(t) - \norm {\phivec_0 \cdot \rvec_0(t)}_1} {\alpha\rmax}.
	\end{equation}}%
\end{lemma}

\begin{proof}
To see this, note that every time a node pushes, its estimate increases by
$\alpha \rmax$,	hence the total cost of Algorithm~\ref{alg:lp} is bounded by:
\shrink{
\begin{align*}
	& \sum_{s \in V_i} \frac {d^{\din}_{\tn 0}(s) \times \bvec_0(s, t)} {\alpha \rmax} \\
	=& \frac {\sum_{s \in V_i} d^{\din}_{\tn 0}(s) \times (\pi_{\tn 0} (s, t)  - \sum_{x \in V_i} \pi_{\tn 0} (s, x) \times \rvec(x, t))} {\alpha\rmax} \\
	=& \frac {\phivec_{\tn 0}(t) - \norm {\phivec_0 \cdot \rvec_0(t)}_1} {\alpha\rmax}
\end{align*}}
\end{proof}

Now consider the $i$-th edge update, for $i = 1,\dots,\K$.
That is, we have had the updated graph $G_i$, after updating $G_{i-1}$ with $e_i = u_i \rightarrow v_i$.
Then we run Algorithm~\ref{alg:lpUpdate} to update $\bvec_{i-1}(t)$ and $\rvec_{i-1}(t)$.
Let
\shrink{\begin{align*}
		&	(\bvec_i(t), \rvec_i(t))\triangleq \\
		&	\quad \mbox{\alglpUpd}	(t, \bvec_{i-1}, \rvec_{i-1}(t), u_i, v_i, G_i, \rmax).
\end{align*}}%
For simplicity, let $\Delta_i(t)$ denote the updated amount of residuals for this edge update.
That is, if we are inserting $e_i$, then from step~\ref{step:ins} of Algorithm~\ref{alg:lpUpdate},
$\Delta_i(u_i, v_i, t)$ is defined to be $1/\alpha$ times:
\shrink{\begin{equation*}\label{eq:l1Change}
	\left|\frac {(1-\alpha) \times \bvec_{i-1}(v_i, t) - \bvec_{i-1}(u_i, t) -
			\alpha \times \rvec_{i-1}(u_i, t) + \alpha \times \mathbf 1_{u_i = t}}
			{d^{\dout}_{i}(u_i)} \right|
\end{equation*}}
and it could be similarly defined for deletion.

Lemma~\ref{lem:amortize} is the heart of our amortized analysis.
The key observation is that one can take care of the cost of an edge update,
by comparing the amount of residuals that we pushed out,
to the amount of new residuals that get created.
Note that the difference between these two masses is precisely the amount of
mass that has been received into the estimates.
Another useful property of push algorithms is monotonicity:
this property has played a crucial role in all the running time analysis
of push algorithms.
As long as we only push positive residuals or only negative residuals,
then the estimates will only change in one way but not the other.
This ensures that we could use estimates as a potential function to
bound the number of push operations a vertex does.

\begin{lemma}\label{lem:amortize}
	The running time of Algorithm~\ref{alg:lpUpdate} for updating $\bvec_{i-1}(t)$ and $\rvec_{i-1}(t)$
	for $G_i$ is at most
		\shrink{
		\begin{align*}
			& \frac {\phivec_{i}(u_i) \times \Delta_i(u_i, v_i, t)} {\alpha\rmax}
			+ \frac {\norm {\phivec_i - \phivec_{i-1}}_1} {\alpha} \\
			&+ \frac {\norm {\phivec_{i-1} \cdot \rvec_{i-1}(t)}_1 - \norm {\phivec_{i} \cdot \rvec_{i}(t)}_1} {\alpha\rmax}
		\end{align*}}
\end{lemma}

\begin{proof}
	Let $\bvec'(t)$ and $\rvec'(t)$ denote the updated values of $\bvec_{i-1}(t)$ and $\rvec_{i-1}(t)$
	after step~\ref{step:upd} in Algorithm~\ref{alg:lpUpdate}.
	Note that when we invoke Algorithm~\ref{alg:lp} to reduce the maximum residual of $\rvec'(t)$,
	we first work on positive residuals and then work on negative residuals.
	We bound the cost of the two phases separately.

	We first bound the cost of pushing out positive residuals (step~\ref{step:pushPos1} and \ref{step:pushPos2} in Algorithm~\ref{alg:lp}).
	Let $\bvec''(t)$ and $\rvec''(t)$ denote the estimate and residual vector after step~\ref{step:pushPos2}. 
	Since only positive residuals are pushed out, $\bvec''(s, t) \ge \bvec'(s, t)$, for any $s \in V_i$.
	Hence the cost of this reverse local push is at most:
	\shrink{
	\begin{equation}\label{eq:updCost}
		 T^+ \triangleq \sum_{s \in V_i} \frac {d^{\din}_{i}(s) \times (\bvec''(s, t) - \bvec'(s, t))} {\alpha\rmax}
	\end{equation}}

	By Equation \eqref{eq:residualForm},
	\shrink{\[ \bvec''(s, t) = \pi_{i}(s, t) - \sum_{x \in V_i} \pi_{i}(s, x) \times \rvec''(x, t)\]}
	Similarly,
	\shrink{\[ \bvec'(s, t) = \pi_{i}(s, t) - \sum_{x \in V_i} \pi_{i}(s, x) \times \rvec'(x, t)\]}
	Hence the difference is:
	\[ \sum_{x \in V_i} \pi_{i}(s, x) \times (\rvec'(x, t) - \rvec''(x, t)) \]
	Apply the above expression into Equation~\eqref{eq:updCost}, we get:
	\shrink{\begin{align*}
		T^+ =& \sum_{s \in V_i} \sum_{x \in V_i} \frac {d^{\din}_{i}(s) \times \pi_{i}(s, x) \times
		(\rvec'(x, t) - \rvec''(x, t))} {\alpha\rmax} \\
		=& \sum_{x \in V_i} \frac {\phivec_{i}(x) \times (\rvec'(x, t) - \rvec''(x, t))} {\alpha\rmax}
	\end{align*}}%
	Here we used $\phivec_i(x)$ to simplify the above expression.
	Now we show that the above expression is at most:
	\shrink{\begin{equation*}
		T^+ \le \frac {\norm {\phivec_{i} \times \rvec'(t)}_1 - \norm {\phivec_{i} \cdot \rvec''(t)}_1} {\alpha\rmax}
	\end{equation*}}%
	To check this, we compare each vertex $x \in V_i$ and their corresponding entries in the above expression.
	Since $\phivec_{i}(x)$ is positive for every $x \in V_i$, if $\rvec''(x, t) \ge 0$, then clearly
	\shrink{\begin{align*}
			&\phivec_{i}(x) \times (\rvec'(x, t) - \rvec''(x, t)) \\
		\le& \phivec_{i}(x) \times (\abs {\rvec'(x, t)} - \abs {\rvec''(x, t)})
	\end{align*}}%
	If $\rvec''(x, t) < 0$, then we infer that $x$ has not performed any push operation:
	otherwise $\rvec''(x, t)$ should be nonnegative.
	Note that $x$ only received positive residual updates during this phase of
	forward push. This shows
	$\rvec'(x, t) \le \rvec''(x, t) < 0$. Therefore:
	\shrink{\begin{align*}
		&\phivec_{i}(x) \times (\rvec'(x, t) - \rvec''(x, t)) \\
		\le & 0 \le \phivec_{i}(x) \times (\abs {\rvec'(x, t)} - \abs {\rvec''(x, t)})
	\end{align*}}

	We then bound the cost of pushing out negative residuals (step~\ref{step:pushNeg1} and \ref{step:pushNeg2} in Algorithm~\ref{alg:lp}).
	Since only negative residuals are pushed out, $\bvec_i(s, t) \le \bvec''(s, t)$, for any $s \in V_i$.
	By the same argument we used to derive Equation \eqref{eq:updCost}, the total cost of this reverse local push is at most:
		\shrink{
		\begin{align}\label{eq:updCostNeg}
			T^- \triangleq \sum_{x \in V_i} \frac {\phivec_i(x) \times (\rvec_i(x, t) - \rvec''(x, t))} {\alpha \rmax}
		\end{align}}
	We claim that:
		\shrink{
		\begin{align}\label{eq:boundCostNeg}
			T^- \le \frac {\norm {\phivec_i \cdot \rvec''(t)}_1 - \norm {\phivec_i \cdot \rvec_i(t)}_1} {\alpha \rmax}
		\end{align}}%
	To see this, we compare each vertex $x \in V_i$	and their corresponding entries
	between Equation \eqref{eq:updCostNeg} and \eqref{eq:boundCostNeg}.
	If $\rvec_i(x, t) \le 0$, then clearly:
		\shrink{
		\begin{align*}
					&\phivec_i(x) \times (\rvec_i(x, t) - \rvec''(x, t)) \\
			\le &\phivec_i(x) \times (\abs {\rvec''(x, t)} - \abs {\rvec_i(x, t)})
		\end{align*}}%
	If $\rvec_i(x, t) > 0$, then we infer that $x$ does not push out any negative residuals,
	otherwise $\rvec_i(x, t)$ will be non-positive.
	This further implies that $\rvec''(x, t) \ge \rvec_i(x, t) > 0$,
	since $x$ may only receive negative residual updates during this phase.
	Therefore,
		\shrink{
		\begin{align*}
			& \phivec_i(x) \times (\rvec_i(x, t) - \rvec''(x, t)) \\
			\le 0 & \le \phivec_i(x) \times (\abs {\rvec''(x, t)} - \abs {\rvec_i(x, t)}))
		\end{align*}}
	
	Finally, we obtain a bound on the total running time of invoking Algorithm~\ref{alg:lp}
	to reduce the maximum residual:
		\begin{align}\label{eq:abs}
			T \triangleq T^+ + T^- \le \norm {\phivec_i \cdot \rvec'(t)}_1 - \norm {\phivec_i \cdot \rvec_i(t)}_1
		\end{align}
	We work on the above equation to finish the proof.
	First, since $\rvec'(t)$ is $\rvec_{i-1}(t)$ with an update of $\Delta_i(u_i, v_i, t)$ on vertex $u_i$,
	we can separate out $\Delta_i(u_i, v_i, t)$ from $\rvec'(t)$:
	\shrink{\begin{equation*}
		\norm {\phivec_{i} \cdot \rvec'(t)}_1 \le \phivec_{i}(u_i) \times \Delta_i(u_i, v_i, t)
		+ \norm {\phivec_{i} \cdot \rvec_{i-1}(t)}_1
	\end{equation*}}
	Applying the above equation into Equation~\eqref{eq:abs}, we found that
	\shrink{\begin{equation}\label{eq:sep1}
		T \le \frac {\norm {\phivec_i \cdot \rvec_{i-1}(t)}_1 + \phivec_i(u_i) \times \Delta_i(u_i, v_i, t) - \norm{\phivec_i \cdot \rvec_i(t)}_1 }
						{\alpha\rmax}
	\end{equation}}
	Then, since $\norm {\rvec_{i-1}(t)}_{\infty}$ is at most $\rmax$,
	\shrink{\begin{align*}\label{eq:sep2}
		\norm {\phivec_{i} \cdot \rvec_{i-1}(t)}_1 
		\le& \norm {(\phivec_i - \phivec_{i-1}) \cdot \rvec_{i-1}(t)}_1
			+	\norm {\phivec_{i-1} \cdot \rvec_{i-1}(t)}_1 \\
		\le& \rmax \times \norm {\phivec_{i} - \phivec_{i-1}}_1
		+ \norm {\phivec_{i-1} \cdot \rvec_{i-1}(t)}_1
	\end{align*}}%
	Applying the above equation into equation \eqref{eq:sep1}
	gives us the desired conclusion.
\end{proof}

Lemma~\ref{lem:amortize} naturally suggests that
there is a way to amortize per edge update costs,
by cancelling out the weighted residual terms.
Hence, by summing up the initialization cost in Lemma~\ref{lem:initCost},
and the edge update cost in Lemma~\ref{lem:amortize}, for $i = 1,\dots,\K$,
we obtained the total cost of maintaining the estimates and residuals
for every graph $G_i$, from $i=0,\dots,\K$:
\shrink{\begin{align*}
	\psi(t) \triangleq &\frac {\phivec_0(t)} {\alpha\rmax}
			+ \sum_{i=1}^{\K} \frac {\phivec_{i}(u_i) \times \Delta_i(u_i, v_i, t)} {\alpha\rmax}
			+ \sum_{i=1}^{\K} \frac {\norm {\phivec_{i} - \phivec_{i-1}}_1} {\alpha}
\end{align*}}%

Now we are ready to present an average case analysis, by summing up $\psi(t)$ over all the target vertices.
First, we know from the work of Lofgren et al. (Theorem $1$, \cite{kdd:fastppr}) that:
\begin{equation*}
	\sum_{t \in V_0} {\phivec_0(t)} = m.
\end{equation*}
Therefore, the total running time of maintaining a pair of estimates and residuals with threshold $\rmax$
for every node $t \in V_0$,
starting from $G_0$ with $m$ edges,
during $\K$ edge updates,
is at most:
\shrink{\begin{align}\label{eq:totalCost}
	\Psi \triangleq \sum_{t \in V_0} \psi(t) =
			\frac {m} {\alpha \rmax}
		+	\sum_{i=1}^{\K} \sum_{t \in V_0} \frac {\phivec_i(u_i) \times \Delta_i(u_i, v_i, t)} {\alpha \rmax}\\
		+ \sum_{i=1}^{\K} \left(\sum_{t \in V_0} \frac {\norm {\phivec_{i} - \phivec_{i-1}}_1} {\alpha} \right)
\end{align}}%
where we have changed the order of summation for the second and third term.

\begin{lemma}\label{lem:sumDelta}
	Let $t$ be any vertex of $V_i$.
	Then \[\sum_{t \in V_i} \Delta_i(u_i, v_i, t) \le (2n\rmax + 2) / (\alpha \times d^{\dout}_i(u_i)),\]
	for any $i = 1,\dots, \K$.
\end{lemma}

\begin{proof}
We deal with the sum of numerators of $\Delta_i(u_i, v_i, t)$ first,
since the denominator of $\Delta_i(u_i, v_i, t)$ does not depend on $t$.
We first take out each term from the absolute value to obtain an upper bound of $\Delta_i(u_i, v_i, t)$:
\shrink{\begin{align*}
	   & (1-\alpha) \times \bvec_{i-1}(v_i, t) + \bvec_{i-1}(u_i, t) \\
		 &	+ \alpha \times \rvec_{i-1}(u_i, t)+ \alpha \times \mathbf 1_{u_i = t} \\
	\le&~(1-\alpha) \times (\pi_{i-1}(v_i, t) + \rmax) + (\pi_{i-1}(u_i, t) + \rmax)  \\
		&		+\alpha \times \rmax + \alpha \times \mathbf 1_{u_i = t}
\end{align*}}%
We used the fact that $\bvec_{i-1}(u_i, t) \le \pi_{i-1}(u_i, t) + \rmax$ (similarly for $v_i$)
and $\rvec_{i-1}(u_i, t) \le \rmax$.
If we sum up the above expression over $t \in V_i$,
and take the denominator back to the expression,
it leads to the following much simplified conclusion:
	\begin{align*}
		\sum_{t \in V_i} \Delta_i(u_i, v_i, t) \le \frac {2n\rmax + 2} {\alpha \times d^{\dout}_i(u_i)}
	\end{align*}
\end{proof}

\iftoggle{fullpaper}{
	At the end, we prove Theorem~\ref{thm:revPush} by solving $\Psi$ in each dynamic graph model.
}

\paragraph{Arbitrary edge update on an undirected graph}

\begin{proofof}{Theorem~\ref{thm:revPush}, Part $1$}
Note that Lemma~\ref{lem:initCost} (initialization cost) holds for undirected graphs as well.
When we update an edge, We need to take into account that we applied
insertion/deletion twice (step \ref{step:ins} and \ref{step:del} in Algorithm~\ref{alg:lpUpdate}),
for each direction of the edge.
This makes a difference when we separate out the difference between
$\rvec'(t)$ and $\rvec_{i-1}(t)$.
Hence, it suffices to add
	\[ \frac{\phivec_i(v_i) \times \Delta_i(v_i, u_i, t)} {\alpha \rmax}\]
into the bound of Lemma~\ref{lem:amortize}:
the rest of the proof holds for undirected graphs.
In conclusion, we found that the total running time of maintaining $\bvec_i(t)$ and $\rvec_i(t)$
is at most:
	\shrink{\begin{align*}
		 \Psi \triangleq \frac m {\alpha\rmax}
					&+ \sum_{i=1}^{\K} \sum_{t \in V_0}
						\frac {\phivec_i(u_i) \times \Delta_i(u_i, v_i, t) + \phivec_i(v_i) \times \Delta_i(v_i, u_i, t)} {\alpha\rmax} \\
					&+ \sum_{i=1}^{\K} \sum_{t \in V_0}
						\frac {\norm {\phivec_i - \phivec_{i-1}}_1} {\alpha}
	\end{align*}}%

Now we claim that with $\phivec_i(x) = d_i(x)$, for any $x \in V_i$ and $i = 0,\dots,\K$.
By Proposition~\ref{prop:reversible},
\shrink{\begin{align*}\label{eq:undPhi}
	\phivec_{i}(x) = \sum_{s \in V_i} d_i(s) \times \pi_i(s, x)
			=	\sum_{s \in V_i} d_i(x) \times \pi_{i}(x, s).
				= d_{i}(x).
\end{align*}}%
Combined with Lemma~\ref{lem:sumDelta}, the second term of $\Phi$ is at most
$\K \times (\frac {4n} {\alpha^2} + \frac 4 {\alpha^2 \rmax})$.

It also follows that $\phivec_i$ is equal to the degree vector of $G_i$,
and the $l_1$ difference between $\phivec_{i}$ and  $\phivec_{i-1}$
is equal to 2, since only the degrees of $u_i$ and $v_i$ changed by one.

To sum up,
\begin{equation*}
	\Psi \le \frac m {\alpha\rmax} + \K \times (\frac {4n} {\alpha^2} + \frac 4 {\alpha^2\rmax})
				+ \frac {2n\K} {\alpha}.
\end{equation*}
\end{proofof}

\iftoggle{fullpaper}{
	\paragraph{Random edge permutation of a directed graph}
		\input{random_permutation_proof}
}

\subsection{Forward push}

Now we apply a similar approach to derive a dynamic forward push algorithm.
Let $s$ be a source vertex of $G$.
Let $\rmax$ be a threshold parameter between $0$ and $1$.
Our goal is to maintain a pair of estimates $\fvec(s)$ and residuals $\rvec(s)$
such that $\gvec(s, t) / d^{\dout}(t) \le \rmax$, for any $t \in V$.
We begin by describing an equivalent invariant property to Equation~\ref{eq:fwdReisdualForm}.
\begin{lemma}\label{lem:invarianceFwd}
	Equation~\ref{eq:fwdReisdualForm} implies
		\[ \fvec(s, t) + \alpha \times \gvec(s, t) = \sum_{x \in N^{\din}(t)} \frac {\fvec(s, x)} {d^{\dout}(x)} + \alpha \times \mathbf 1_{t = s}, \forall t \in V,\]
	and vice versa.
\end{lemma}
\begin{proof}
	Let $\pivec_s = \ppr s {\cdot}$ denote a vector with personalized PageRank from $s$ to every node of $G$.
	Let $\Pi = \alpha \times (I - (1-\alpha)A^\intercal D^{-1})^{-1}$:
	the fact that $\Pi$ exists follows from Definition~\ref{eq:ppr},
	and $\pivec_s$ is equal to $\Pi \cdot \vec e_s$.
	Therefore equation~\ref{eq:fwdReisdualForm} in vector form is equivalent to:
	\shrink{\begin{align*}
										&~\pivec_s = \fvec(s) + \Pi \cdot \gvec(s) \\
		\Leftrightarrow &~\Pi^{-1} \cdot \pivec_s = \Pi^{-1} \fvec(s) + \gvec(s) \\
		\Leftrightarrow &~\vec e_s = \frac {I - (1-\alpha)A^\intercal D^{-1}} {\alpha} \cdot \fvec(s) + \gvec(s) \\
		\Leftrightarrow &~\fvec(s) + \alpha \times \gvec(s) = (1-\alpha) A^\intercal D^{-1} \fvec(s) + \alpha \times \vec e_s
	\end{align*}}
\end{proof}

Now we use Lemma~\ref{lem:invarianceFwd} to derive the update procedure.
Consider when an edge $u \rightarrow v$ is inserted to $G$.
Since the outdegree of $u$ increases by $1$, the invariant does not hold
for the out neighbors of $u$ any more:
the reason being that every out neighbor receives an equal proportion of
mass from $u$ which is $(1-\alpha)$ divided by the outdegree of $u$ times
the amount of mass that $u$ pushed.
To solve this imbalance,
clearly a simple solution is to scale $\fvec(s, u)$ by a factor of $(d^{\dout}(u) + 1) / d^{\dout}(u)$.
This ensures the invariant for every out neighbor of $u$, except $v$.
This is because we need to take into account the amount of residuals $v$ 
should have received previously, had the edge $u \rightarrow v$ existed.
Finally, since we have increased $\fvec(s, u)$, this will break the invariant
for $u$. Hence we will reduce $\gvec(s, u)$ by the increased amount on
$\fvec(s, u)$, divided by $\alpha$.
See Algorithm~\ref{alg:lpUpdateFwd} below for details.
To work with an undirected graph, we will apply insert/delete twice,
for both direction between $u$ and $v$.
It's not hard to see this, following our discussion above.

\newcommand{\alglpFwdUpd}{{\sc UpdateForwardPush}}
\begin{algorithm}
	\small
	\caption{\alglpFwdUpd}
	\label{alg:lpUpdateFwd}
	\begin{algorithmic}[1]
		\Input ($s, \fvec(s), \gvec(s), u, v, G, \rmax$)
		\Require $G$ is a directed graph. Let $u \rightarrow v$ be the previous edge update,
						with $G$ being the updated graph.
		\State Apply Insert/Delete to $\fvec(s)$ and $\gvec(s)$.
		\State \Return \alglpFwd$(s, \fvec(s), \gvec(s), G, \rmax)$
		\vspace{0.05in}
		\Procedure{Insert}{$u, v$} \label{step:insertFwd}
			\State $\fvec(s, u) \mathrel{*}= \frac {d^{\dout}(u)} {d^{\dout}(u) - 1} $\label{step:fwdIns1}
			\State $\gvec(s, u) \mathrel{-}= \frac {\fvec(s, u)} {d^{\dout}(u)} \cdot \frac 1 {\alpha} $ \label{step:fwdIns3}
			\State $\gvec(s, v) \mathrel{+}= \frac {(1-\alpha) \times \fvec(s, u)} {d^{\dout}(u)} \cdot \frac 1 {\alpha} $\label{step:fwdIns2}
		\EndProcedure
		\vspace{0.05in}
		\Procedure{Delete}{$u, v$}
			\State $\fvec(s, u) \mathrel{*}= \frac {d^{\dout}(u)} {d^{\dout}(u) + 1}$\label{step:fwdDel1}
			\State $\gvec(s, u) \mathrel{+}= \frac {\fvec(s, u)} {d^{\dout}(u)} \cdot \frac 1 {\alpha} $
			\State $\gvec(s, v) \mathrel{-}= \frac {(1-\alpha) \times \fvec(s, u)} {d^{\dout}(u)} \cdot \frac 1 {\alpha} $
		\EndProcedure
		\newline
	\end{algorithmic}
\end{algorithm}

Next we prove a theorem on the update cost of Algorithm~\ref{alg:lpUpdateFwd}. 
We will prove it on undirected graphs.
It's not clear to us how to analyze directed graphs with forward push:
the difficulty being that there is no clean bound on the error of Algorithm~\ref{alg:lpFwd},
between the estimates and the true personalized PageRank value.

\vspace{-0.05in}

\begin{theorem}\label{thm:fwdPush}
	Let $\seq {G_i = (V_i, E_i)}$ be a sequence of $\K+1$ undirected graphs,
	such that each graph is obtained from the previous graph by one edge update.
	Let $s$ be a random vertex of $V_0$.
	And let $\rmax$ be a parameter between $0$ and $1$.
	Then the total running time of maintaining a forward local push solution $\fvec_i(s)$ for each
	graph $G_i$ such that
	$\abs {\fvec_i(s, t) - \pi_i(s, t)} / d_i(t) \le \rmax$,
	for any $t \in V_i$,
	using Algorithms \ref{alg:lpUpdateFwd} is at most $O(\K / (\alpha^2) + \K / (n \alpha^2 \rmax) + 1 / (\alpha\rmax))$.
\end{theorem}

As long as $\abs {\gvec_i(s, t)} / d_i(t) \le \rmax$,
for any $t \in V_i$ and $i = 0,\dots,\K$,
then $\abs {\fvec_i(s, t) - \pi_i(s, t)} / d_i(t) \le \rmax$.
This accuracy guarantee follows from the work of Anderson et al. and Lofgren et al.~\cite{focs:chung, thesis:lofgren}.
Hence, we will focus on analyzing running time.
To derive this theorem, we divide the arguments into three parts:
first, we present a bound on the initiation cost as well as the update cost per edge;
secondly, we amortize the costs together, cancelling out the residual terms;
finally, we use properties from undirected graphs to bound the total cost.
\iftoggle{fullpaper}{
	\input{proof_fwd_push}

}{
	\noindent The proof can be found in the full version of this paper.
}

\subsection{Discussions}\label{sec:discussions}

The algorithms presented in the previous sections do not handle dangling nodes.
We consider how to take care of this situation.
We also extend the algorithms to handle newly arrived nodes.

\vspace{-0.1in}

\paragraph{Dangling nodes}
There are two possible ways to handle dangling nodes.
The first way is to create a separate sink node.
Consider the case of doing reverse push (similarly for forward push).
If a node $u$ does not have any in neighbors and need to perform a push operation,
then the amount of mass will be pushed to the sink node.
These mass will stay at the sink node.
Another way is to insert an edge from $t$ to $u$,
if $t$ is the target node for maintaining reverse push solutions.
Later on if an edge is added from $v$ to $u$,
so that $u$ is not a dangling node anymore,
then we first delete the edge from $t$ to $u$ and then insert an edge from $v$ to $u$.

\vspace{-0.05in}

\paragraph{Node arrivals}
When a new node arrives, one can insert the edges one by one
using the procedures in Algorithm~\ref{alg:lpUpdate} (line \ref{step:insertBk}) and Algorithm~\ref{alg:lpUpdateFwd} (line \ref{step:insertFwd}),
and then invoke the local push algorithms to reduce the maximum absolute values of residuals.
Similarly for node deletions.

%% file: random_permutation_proof.tex
\begin{lemma}\label{lem:rpDelta}
	Let $e_i = u_i \rightarrow v_i$ denote the $i$-th arrived edge in the random edge permutation, then
	$\E { \frac {\phivec_{i}(u_i)} {d^{\dout}_{i}(u_i)} } = 1$.
\end{lemma}
\begin{proof}
	Without loss of generality, we can assume that $E_i$ is the first $m+i$ edges in the random edge permutation,
	and then take expectation over a random permutation of $E_i$.
	Once we have proved the Lemma for this case, then by the linearity of expectations,
	we can prove it over the random permutation of all the edges.
	From the work of Bahmani et al.~\cite{vldb:walkUpd},
	we know that the probability that $u_i$ is the start vertex of $e_i$ is equal to
	$d^{\dout}_{i}(u_i) / (m + i)$.
	Therefore,
	\begin{equation}\label{eq:deltaDirct}
		\E { \frac {\phivec_{i}(u_i)} { d^{\dout}_{i}(u_i) } }
		= \sum_{x \in V_i} \frac {\phivec_{i}(x)} {m + i} = 1
	\end{equation}
	Note in particular that $\phivec_i(\cdot)$ is fixed once we have fixed $E_i$,
	and the last equation follows from Theorem 1 of Lofgren et al.~\cite{kdd:fastppr}.
\end{proof}

In the following Lemma, we show that the expected change from $\phivec_i$ to $\phivec_{i-1}$
is ``small''.

\begin{lemma}\label{lem:rpPur}
	Let $e_i = u_i \rightarrow v_i$ denote the $i$-th arrived edge, then
	$\E {\sum_{t \in V_i} \norm {\phivec_{i} - \phivec_{i-1}}_1} \le 1 + 4 /\alpha$.
\end{lemma}

\begin{proof}
	We first put in the definition of $\phivec_i$ and $\phivec_{i-1}$,
	\shrink{\begin{align*}
		T \triangleq &\sum_{t \in V_i} \norm {\phivec_{i} - \phivec_{i-1}} \nonumber \\
		=& \sum_{t \in V_i} \sum_{s \in V_i} \abs { d^{\din}_{i}(s) \times \pi_{i}(s, t) - d^{\din}_{i-1}(s) \times \pi_{i-1}(s, t) }
	\end{align*}}%
	Since only $v_i$'s indegree increased by 1, and the indegrees of other vertices do not change,
	we can take out this additional term of $v_i$, and get an upper bound:
	\shrink{\begin{align}
		 T \le &\sum_{t \in V_i} \pi_{i}(v_i, t) + \sum_{s \in V_i} \sum_{t \in V_i} \abs {d^{\din}_{i-1}(s) \times (\pi_{i}(s,t) - \pi_{i-1}(s, t))} \nonumber\\
		=& 1 + \sum_{s \in V_i} d^{\din}_{i-1}(s) \times \left( \sum_{t \in V_i} \abs {\pi_{i}(s, t) - \pi_{i-1}(s, t)} \right) \label{eq:rpC}
	\end{align}}%
	We argue that in the above expression,
		\[ \sum_{t \in V_i} \abs{\pi_{i}(s, t) - \pi_{i-1}(s, t)} \le \frac {2 \times \pi_{i-1}(s, u_i)} {\alpha \times d^{\dout}_{i}(u_i)} \]
	To see this, think of the random walk definition of personalized PageRank.
	The left hand side is at most twice times the probability that a random walk
	from $s$ needs to be updated after inserting $e_i$:
	this follows from the work of Bahmani et al.~\cite{vldb:walkUpd}.
	And we have used Proposition~\ref{prop:walkUpd} to get this probability.
	Therefore equation~\eqref{eq:rpC} becomes:
	 \shrink{\begin{align*}
	 		1 + \sum_{s \in V_i} \frac {2 \times d^{\din}_{i-1}(s) \times \pi_{i-1}(s, u_i) } {\alpha \times d^{\dout}_{i}(u_i)}
	 &= 1 + \frac {2\times \phivec_{i-1}(u_i)} {\alpha \times d^{\dout}_{i}(u_i)} \\
	 &\le 1 + \frac {4 \times \phivec_{i-1}(u_i)} {\alpha \times d^{\dout}_{i-1}(u_i)}
	 \end{align*}}%
	Here we used that assumption that $d_{i-1}(u_i)$ is not zero:
	since we assumed that there are no dangling nodes in $G_0$,
	there will not be dangling nodes after edge insertions.
	By a similar argument to that of Lemma~\ref{lem:rpDelta}:
	\begin{equation*}
		\E { \frac {\phivec_{i-1}(u_i)} {d^{\dout}_{i-1}(u_i)} } = 1.
	\end{equation*}
\end{proof}

\begin{proofof}{Theorem~\ref{thm:revPush}, Part $2$}
We first apply Lemma~\ref{lem:sumDelta}
to the total running time $\Psi$ we derived in Equation~\eqref{eq:totalCost} and obtain:
\begin{align*}
	\E {\Psi} \le \frac {m} {\alpha \rmax}
					+ \sum_{i=1}^{\K} \E { \frac {\phivec_i(u_i) \times (2n\rmax + 2)} {\alpha^2 \rmax \times d^{\dout}_{i}(u_i)} } \\
					+ \sum_{i=1}^{\K} \E {\sum_{t \in V_0} \norm {\phivec_i - \phivec_{i-1}}_1} / \alpha
\end{align*}

By Lemma~\ref{lem:rpDelta} and Lemma~\ref{lem:rpPur} (note that $V_0 \subseteq V_i$),
the above expression is at most:
\begin{equation}\label{eq:dirRev}
	\E{\Psi} \le \frac m {\alpha\rmax} +
			\K \times (\frac {2n} {\alpha^2} + \frac 2 {\alpha^2\rmax}) +
			\K \times (\frac 1 {\alpha} + \frac 4 {\alpha^2}).
\end{equation}
\end{proofof}

%% file: proof_fwd_push.tex
Let
	{\begin{align*}
	 (\fvec_0(s), \gvec_0(s)) \triangleq \mbox{\alglpFwd}(s, \vec e_s, \mathbf 0, G_0, \rmax)
	\end{align*}}%

\begin{lemma}\label{lem:initFwd}
	The running time of Algorithm~\ref{alg:lp} is at most:
	\begin{equation*}
		\frac {1 - \sum_{t \in V_0} \gvec_0(s, t)} {\alpha\rmax}
	\end{equation*}
\end{lemma}
\begin{proof}
Note that every time $t$ pushes, it's estimate increases by at least $\alpha \rmax d^{\dout}(t)$,
therefore total cost is at most:
	\shrink{\begin{align*}
		 & \sum_{t \in V_0} \frac {\fvec_0(s, t)} {\alpha \rmax \times d_0(t)} \times d_0(t)
		= \sum_{t \in V_0} \frac {\fvec_0(s, t)} {\alpha \rmax}\\
		=& \sum_{t \in V_0} \frac {\pi_0(s, t) - \sum_{x \in V_0} \gvec_0(s, x) \times \pi_0(x, t)} {\alpha \rmax} \\
		=& \frac {\sum_{t \in V_0} \pi_0(s, t) - \sum_{x \in V_0} \sum_{t \in V_0} \gvec_0(s, x) \times \pi_0(x, t)  } {\alpha \rmax} \\
		=& \frac {1 - \sum_{x \in V_0} \gvec_0(s, x)} {\alpha\rmax}
	\end{align*}}%
\end{proof}

Then we derive a result of the update time per edge. Let
\shrink{\begin{align*}
	 &(\fvec_{i}(s), \gvec_{i}(s)) \triangleq \\
	 &		\quad \mbox{\alglpFwdUpd}(s, \fvec_{i-1}(s), \gvec_{i-1}(s), u_i, v_i, G_i, \rmax)
\end{align*}}%
for each $i = 1,\dots,\K$.
And let $\Delta_i(s)$ denote the updated residual amount.
That is, if $e_i$ is an insertion, then by taking absolute values from step~\ref{step:fwdIns3} and~\ref{step:fwdIns2},
we obtain
\shrink{\begin{align*}
	\Delta_i(s) \triangleq \frac {2 - \alpha} {\alpha} \times \frac {\fvec_{i-1}(s, u_i)} {d(u_i)}
\end{align*}}
And similarly if $e_i$ is a deletion.

\begin{lemma}\label{lem:updateFwd}
	Let $e_i = (u_i, v_i)$ denote the $i$-th edge update.
	The running time of Algorithm~\ref{alg:lpUpdateFwd} for updating $\fvec_{i-1}(s)$ and
	$\gvec_{i-1}(s))$ is at most:
	\begin{equation}
		\frac {\norm {\gvec_{i-1}(s)}_1 - \norm {\gvec_i(s)}_1 + \Delta_i(s)} {\alpha \rmax}
	\end{equation}
	for each $i = 1,\dots,K$.
\end{lemma}

\begin{proof}
	The same as Lemma~\ref{lem:amortize}.
\end{proof}

\begin{proofof}{Theorem~\ref{thm:fwdPush}}
	Combining Lemma~\ref{lem:initFwd} and~\ref{lem:updateFwd},
	we conclude that the total running time of maintaining
	a forward local push solution for $s$ is at most:
	\shrink{\begin{equation}\label{eq:costFwd}
		\psi(s) \triangleq \frac {1 + \sum_{i = 1}^{\K} \Delta_i(s)} {\alpha \rmax}
	\end{equation}}

	We know that $\fvec_{i-1}(s, u_i) \le \pi_{i-1}(s, u_i) + \rmax \times d_{i-1}(u_i)$:
	this follows from the accuracy guarantee of $\fvec_{i-1}(s, u_i)$.
	If $e_i$ is an edge insertion, then:
	\shrink{\begin{align}\label{eq:costEd}
		\Delta_i(s) =& \frac {\fvec_{i-1}(s, u_i)} {d_{i}(u_i)} \times \frac {2 - \alpha} {\alpha} \\
					\le& \frac {\pi_{i-1}(s, u_i) + \rmax \times d_{i-1}(u_i)} {d_i(u_i)} \times \frac {2 - \alpha} {\alpha} \\
					\le& \frac {\pi_{i-1}(s, u_i) + \rmax \times d_{i-1}(u_i)} {d_{i-1}(u_i)} \times \frac {2 - \alpha} {\alpha}
	\end{align}}%
	By Proposition~\ref{prop:updateBound},
		\[\sum_{s \in V_i} \frac {\pi_{i-1}(s, u_i)} {d_{i-1}(u_i)} \le 1.\]
	Therefore,
	 \[ \Psi \triangleq \sum_{s \in V_0} \psi(s) \le \frac n {\alpha\rmax} + \K \times \frac {1 + n \rmax} {\alpha\rmax} \times \frac {2 - \alpha} {\alpha}\]
	If $e_i$ is an edge deletion, it's clear that one can use the right expression
	for $\Delta_i(s)$ to obtain a similar conclusion.
	Details omitted.
\end{proofof}

%% file: experiment.tex
\section{Numerical results}

\newcommand{\algLazyFwdPush}{{\sc LazyFwdUpdate}}
\newcommand{\algTracking}{{\sc TrackingPPR}}
\newcommand{\algRandwalk}{{\sc RandomWalk}}

In this section, we evaluate our approach in experiments.
We first compare our dynamic forward local push algorithm (Algorithm~\ref{alg:lpUpdateFwd}) to the previous
work of Ohsaka et al.~\cite{kdd:dynPr}.
Our dynamic Forward Push algorithm differs from Ohsaka et. al.'s in an important way.  When an edge $(u, v)$ arrives, they propose to push the change in residual immediately to all of $u$'s neighbors, requiring $\Omega(d(u))$ time, in addition to any pushes needed to restore the invariant.  In contrast, we modify the estimate and residual values only at $u$ and $v$, taking only $O(1)$ time, before performing any pushes needed to restore the invariant.  We found that this simple optimization decreased the number of residual values updated and led to a 1.5 - 3.5 times improvement in running time, without sacrificing accuracy.

We then make a comparison to the random walk approach of Bahmani et al.~\cite{vldb:walkUpd}.
We consider the problem of identifying the top-K vertices that have the highest personalized PageRank
from a source node,
with random walks being commonly used to solve it on social networks~\cite{gupta2013wtf,vldb:frogwild}.
We found that the random walk algorithm uses 4.5 to 12 times as much storage compared to forward local push algorithm,
and requires 1.5 to 2.5 times as much time to update as well,
to achieve the same level of precision.

Finally, we evaluate our dynamic reverse push algorithm.
We compare against the baseline where we recompute the reverse push from scratch after every edge update.
We found that our approach is 100x faster than this baseline.

\subsection{Experimental setup}

We implement the experiments in Scala.
We use an Amazon EC2 Ubuntu machine with 64GB of RAM
and 16 processors of Intel(R) Xeon(R) CPU E5-2676 v3 @ 2.40GHz.
We tested with three undirected social networks graphs, all downloaded from \url{https://snap.stanford.edu/data/}.
The statistics are listed in Table~\ref{table:stat} below.
For each social network, we generate a uniformly random edge stream of the graph.
Given a test vertex,
we initialize its data structures using the subgraph that consists of the first half of the edges in the random edge stream.
After that, for each new edge arrival, we update its data structure depending the approach we are using.
The teleport probability $\alpha$ is set to $0.2$ in all the experiments.
For ease of comparison, we use \algLazyFwdPush~
to refer to Algorithm~\ref{alg:lpUpdateFwd} (because it ``lazily'' avoids pushing from vertices that receive new edges),
Ohsaka et al.'s Algorithm \algTracking~
and the random walk update Algorithm by~\algRandwalk.

\begin{table}[!ht]
	\centering
	\begin{tabular}{| l | l | l | l |}
		\hline
		graph & \# nodes & \# edges \\
		\hline
		DBLP & 317,080 & 1,049,866 \\
		\hline
		LiveJournal & 3,997,962 & 34,681,189 \\
		\hline
		Orkut & 3,072,441 & 117,185,083 \\
		\hline
	\end{tabular}
	\caption{Basic statistics of graphs that we experimented with.}
	\label{table:stat}
\end{table}

\subsection{Residual updates and running time}

We compare \algLazyFwdPush~and \algTracking~over 100 uniformly sampled source vertices from each graph.
In our implementation,
we first update the previous estimate vectors and residual vectors
based on \algLazyFwdPush~and \algTracking, respectively.
We then invoke the same forward local push routine (Algorithm~\ref{alg:lpFwd})
to reduce any residual that is outside the desired threshold.
At the end of the edge stream,
we compare the performance of each algorithm,
using measurements described below:
\begin{itemize}
	\setlength\itemsep{0.05em}
	\item Residual update: the number of times that a vertex's residual is updated.
		Each residual update corresponds to an update in the priority queue of residuals.
	\item Push iterations: the number of times that Algorithm~\ref{alg:lpFwd}
		performs a push operation in step~\ref{step:fwdPush}.
		The cost of a push iteration is the degree of the pushed vertex times the cost of a residual
		update.
	\item Running time: the amount of update time it takes to maintain estimates and residuals.
		We exclude the amount of time it takes to initialize the data structure.
		Note that this cost is negligible compared to the amount of update time in our experiment.
	\item $l_1$ error: the $l_1$ distance between the computed estimates and the benchmark values.
		To get the benchmark, we ran the local push algorithm with threshold $0.02 / n$,
		where $n$ is the number of vertices of the graph.
	\item Storage: the number of nonzero estimates and residuals maintained.
\end{itemize}

For residual update, push iteraions and running time,
we computed the average value over the 100 samples.
For $l_1$ error, we compute the median error over the 100 samples.
The test results is shown in Table~\ref{table:random}.
We found that \algTracking~does more residual updates than \algLazyFwdPush,
for all three graphs, and the improvement is more significant for denser graphs.
While Algorithm~\ref{alg:lpUpdateFwd} does more push operations than \algTracking,
since the latter already does a push operation before invoking Algorithm~\ref{alg:lpFwd},
the compound effect is that we achieved 1.5 to 5 times speed up, without sacrificing accuracy.

\begin{table*}[t]
	\centering
	\resizebox{0.7\hsize}{!}{%
	\begin{tabular}{|l | ll | ll | ll|}
	\hline
		& \multicolumn{2}{|c|}{DBLP} & \multicolumn{2}{|c|}{LiveJournal} & \multicolumn{2}{|c|}{Orkut} \\
		& {\sc LazyFwd} & {\sc Tracking} & {\sc LazyFwd} & {\sc Tracking} &{\sc LazyFwd} & {\sc Tracking} \\
		& {\sc Update} & {\sc PPR} & {\sc Update} & {\sc PPR} &{\sc Update} & {\sc PPR} \\
	\hline
	Residual updates 	& $1.5 \times 10^6$ & $2.4 \times 10^6$	& $2.5 \times 10^6$ & $1.8 \times 10^7$		& $2.9\times 10^6$ & $6.2 \times 10^7$	\\
	Push iterations 	& $2.3\times 10^5$	&	$1.8 \times 10^5$	&	$1.5 \times 10^5$ & $1.0 \times 10^5$		& $5.3 \times 10^4$ & $2.6\times 10^4$ 	\\
	Running time 			& $2.9s$					&	$4.7s$							&	$17.1$ & $51.4$													& $49.8$ 	& $176.4$ 											\\
	$l_1$ error				&	$0.031$	& $0.031$											&	$0.145$ & $0.150$												&							$0.238$ & $0.244$ \\
	\hline
	\end{tabular}%
	}
	\caption{Test results between \algLazyFwdPush~with $\rmax = 7 \times 10^{-7}$ and
		\algTracking~with $\rmax = 10^{-6}$.
		For each graph, we sampled 100 vertices uniformly at random.}
	\label{table:random}
\end{table*}

\begin{table*}[!ht]
	\centering
		\resizebox{0.7\hsize}{!}{%
	\begin{tabular}{|l | ll | ll | ll|}
	\hline
		& \multicolumn{2}{|c|}{DBLP} & \multicolumn{2}{|c|}{LiveJournal} & \multicolumn{2}{|c|}{Orkut} \\
		& {\sc LazyFwd} & {\sc Random} & {\sc LazyFwd} & {\sc Random} &{\sc LazyFwd} & {\sc Random} \\
		& {\sc Update} & {\sc Walk} & {\sc Update} & {\sc Walk} &{\sc Update} & {\sc Walk} \\
	\hline
	Storage				&	$2.6 \times 10^4$ & $3.2 \times 10^5$	& $5.3 \times 10^4$ & $3.2\times 10^5$ & $6.3\times 10^4$ & $3.2\times 10^5$ \\
	Accuracy			& 0.96	& 0.92	& 0.88 & 0.88 & 0.74 & 0.78 \\
	Running time & 0.32s	& 0.88s	& 8.8s & 15.3s & 29.7s & 53.0s \\ 
	\hline
	\end{tabular}%
	}
	\caption{Test results for top-50 between \algLazyFwdPush~with $\rmax = 5\times 10^{-5}$ and \algRandwalk~with 16000 walks.
		For each graph, we sampled 100 vertices uniformly at random.
		The parameters are chosen so that the median accuracy of both algorithms are around 0.9 on LiveJournal.}
	\label{table:top50}
\end{table*}

\begin{table*}[!ht]
	\centering
	\resizebox{0.7\hsize}{!}{%
	\begin{tabular}{|l | ll | ll | ll|}
	\hline
		& \multicolumn{2}{|c|}{DBLP} & \multicolumn{2}{|c|}{LiveJournal} & \multicolumn{2}{|c|}{Orkut} \\
		& {\sc LazyFwd} & {\sc Random} & {\sc LazyFwd} & {\sc Random} &{\sc LazyFwd} & {\sc Random} \\
		& {\sc Update} & {\sc Walk} & {\sc Update} & {\sc Walk} &{\sc Update} & {\sc Walk} \\
	\hline
	Storage				&	$6.0 \times 10^4$ & $5.0 \times 10^5$	& $1.2 \times 10^5$ & $5.0\times 10^5$ & $1.5\times 10^5$ & $5.0\times 10^5$ \\
	Accuracy			& 0.96	& 0.91	& 0.91 & 0.88 & 0.75 & 0.80 \\
	Running time & 0.41s	& 1.24s	& 8.6s & 16.0s & 30.1s & 59.7s \\
	\hline
	\end{tabular}%
	}
	\caption{Test results for top-100 between \algLazyFwdPush~with $\rmax = 2\times 10^{-5}$ and \algRandwalk~with 25000 walks.
		For each graph, we sampled 100 vertices uniformly at random.
		The parameters are chosen so that the median accuracy of both algorithms are around 0.9 on LiveJournal.}
	\label{table:top100}
\end{table*}

\begin{figure*}[!ht]
	\centering
	\begin{subfigure}[t]{0.5\textwidth}
		\centering
		\includegraphics[width=5cm]{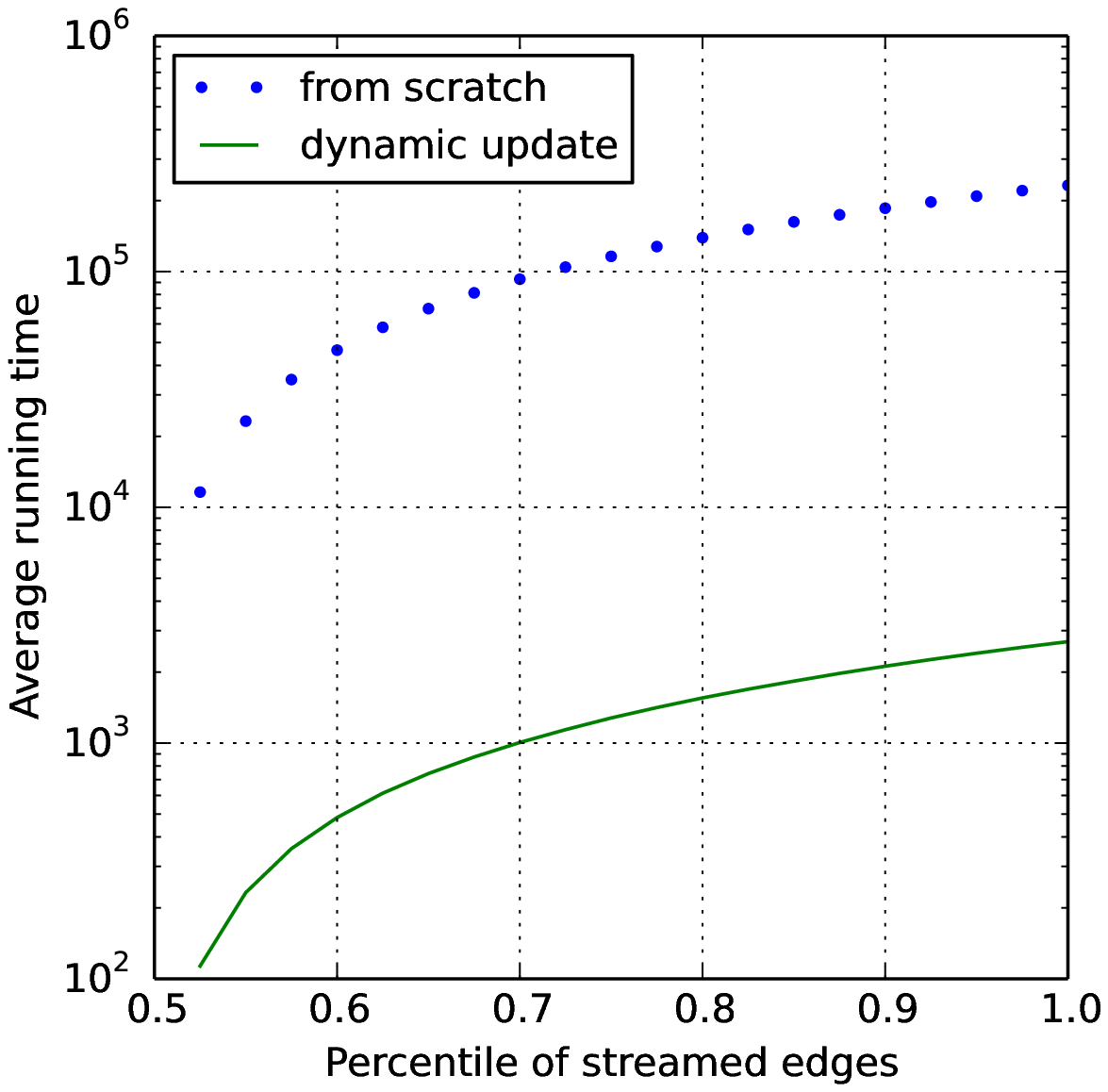}
	\end{subfigure}%
	~
	\begin{subfigure}[t]{0.5\textwidth}
		\centering
		\includegraphics[width=5cm]{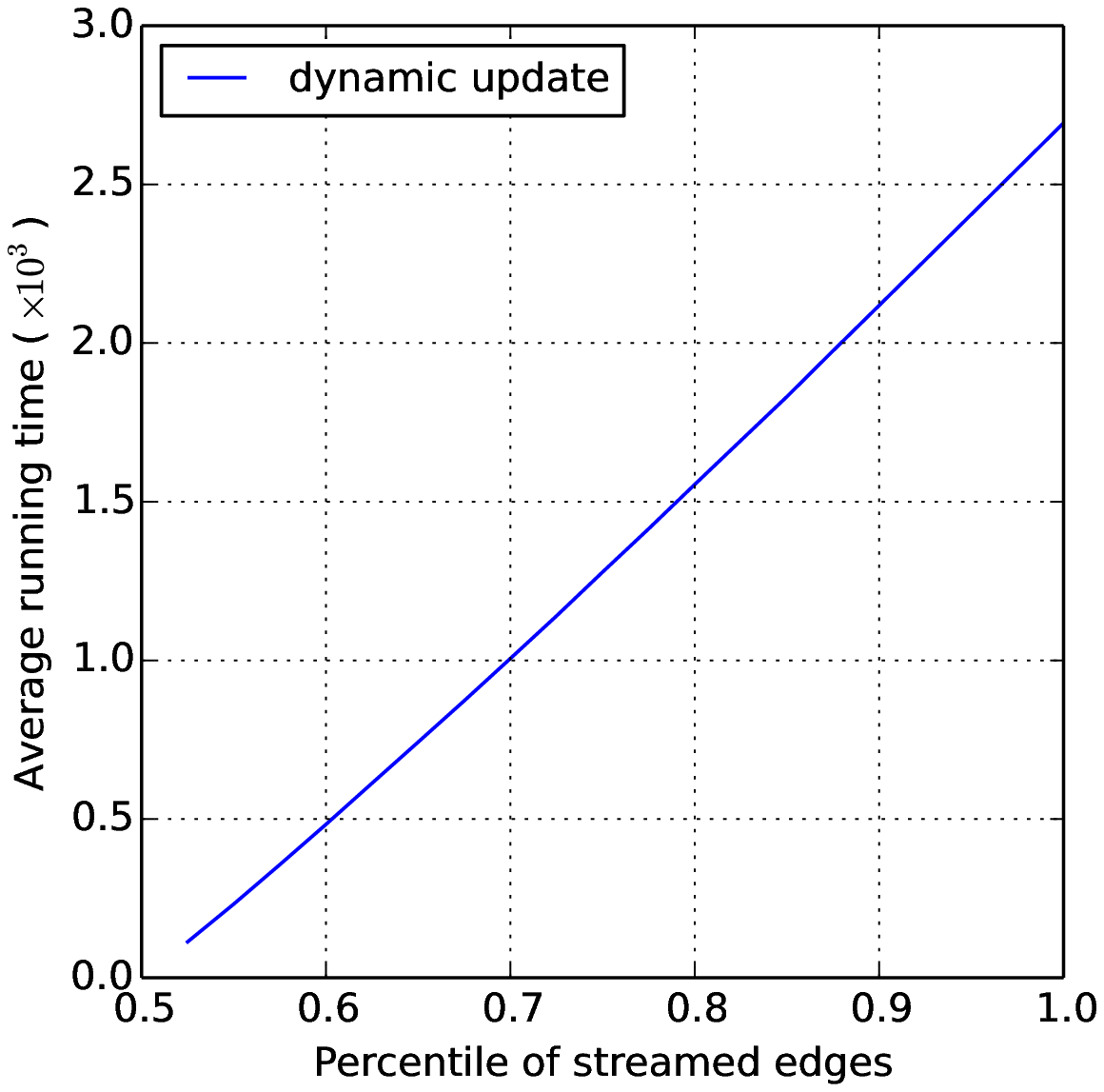}
	\end{subfigure}
	\caption{Running time comparisons on LiveJournal, between maintaining reverse push using Algorithm~\ref{alg:lpUpdate}
		and computing reverse push from scratch for every edge insertion with $\rmax = 10^{-4}$.
		We take 100 random test nodes and compute the average running time for each algorithm.}
	\label{fig:rev_lp}
\end{figure*}

\subsection{Comparison to random walks}

We compare \algLazyFwdPush~and \algRandwalk~for the problem of finding top-K vertices ranked by their
personalized Pagerank from a source vertex.
We sample 100 test vertices uniformly at random.
For each test vertex, the correct top-K solution is the set of K vertices with highest personalized
Pagerank from that vertex.
The performance of each algorithm is measured by:
\begin{itemize}
	\setlength\itemsep{0.05em}
	\item Accuracy: the number of correctly identified vertices (among
		its computed K vertices with highest personalized Pagerank) divided by K.
	\item Running time: the amount of time to maintain/update the data structures.
	\item Storage: to compare storage used by \algRandwalk, we multiply the number of random walks
		by the expected length of a walk ($5$ in our experiment) times 4 (number of bytes to store an integer vertex ID).
		\footnote{To allow for efficient update, it's necessary to build an inverted index for each vertex
		that quickly finds to the set of random walks crossing it.
		In our implementation we built a hashmap for this purpose,
		resulting in an additional factor of $2$ in the amount of storage used.
		Since there might be more efficient implementation, we do not take this additional factor into account
		in our comparison.}
		For \algLazyFwdPush,
		we multiply the number of nonzero estimates and residuals by 8 (number of bytes to store a pair of integer and floating point number).
\end{itemize}
Table~\ref{table:top50} and \ref{table:top100} below describe the results with $K$ being 50 and 100, respectively.
We take the average of running time and storage,
and median of accuracy over the 100 sampled vertices.
For the problem of identifying top-50 nodes,
we found that \algRandwalk~uses 4.5 to 12 times as much storage compared to \algLazyFwdPush,
and uses 1.6 to 2.7 as much running time.
On Orkut graph, when accuracy is controlled to be at the same level as \algRandwalk~for \algLazyFwdPush~(with $\rmax = 4 \times 10^{-5}$),
the average amount of storage becomes $7.0 \times 10^4$ and the average running time becomes $33.3s$.
When $K$ is 100, the test results are qualitatively consistent with the above conclusion.

\vfill\eject

\subsection{Evaluation of dynamic reverse push algorithm}

In this section we evaluate the performance of our dynamic reverse push algorithm.
We compare our proposed solution (Algorithm~\ref{alg:lpUpdate}) to the baseline,
where we recompute reverse push from scratch after every edge insertion using Algorithm~\ref{alg:lp}.
We sample 100 target vertices uniformly at random and take the average of the running time
over these 100 samples.
We have chosen the maximum residual parameter $\rmax = 10^{-4}$ so that at the end of
the edge stream, the median $l_1$ error over the 100 samples between the maintained reverse push estimates
and the true values is less than $0.1$.
Because it takes too long to compute reverse push from scratch after every edge insertion for 100 sampled vertices,
we only computed reverse push from scratch for the first 1000 edge insertions.
We take the average running time for performing reverse push during 1000 edge updates,
and use this average value times the number of inserted edges as an approximation
to the true running time, if we have computed reverse push from scratch for every edge insertion.
Figure~\ref{fig:rev_lp} shows the experiment result on the LiveJournal graph.
We found that our approach is 100x faster than the baseline.
The aggregate running time for updating reverse push has a linear correlation with the number of edges added.
This is consistent with the theoretical bound from Theorem~\ref{thm:revPush}.